\documentclass[11pt]{article}

\usepackage{algorithm}
\usepackage[noend]{algpseudocode}
\usepackage{amsmath, amssymb, amsthm}
\usepackage{thmtools}
\usepackage{fullpage}
\usepackage{thm-restate}
\usepackage{blindtext}
\usepackage{hyperref} 
\usepackage{cleveref}
\usepackage{comment}
\usepackage{enumerate}
\usepackage{mathtools}
\usepackage{subcaption}
\usepackage[dvipsnames]{xcolor}

\colorlet{XXX}{Emerald} \colorlet{mylinkcolor}{XXX}
\colorlet{mycitecolor}{XXX}
\colorlet{myurlcolor}{XXX}
\hypersetup{
  colorlinks = true,
  linkcolor  = blue, 
  citecolor  = blue, 
  urlcolor   = blue, 
}
\crefformat{equation}{(#2#1#3)}

\newcommand{\SEQ}{\textsc{Seq}}
\newcommand{\MSET}{\textsc{MSet}}

\renewcommand{\P}{\mathbb{P}}
\newcommand{\E}{\mathbb{E}}
\newcommand{\Var}{{\normalfont\text{Var}}}  

\newcommand{\Z}{\mathbb{Z}}
\newcommand{\C}{\mathbb{C}}

\newcommand{\R}{\mathbb{R}}
\newcommand{\GEO}{\normalfont\text{Geometric}}

\newcommand{\POIS}{\normalfont\text{Poisson}}
\newcommand{\NB}{\normalfont\text{NegativeBinomial}}
\newcommand{\DEF}{=}
\algdef{SE}[DOWHILE]{Do}{doWhile}{\algorithmicdo}[1]{\algorithmicwhile\ #1}%

\newcommand{\todox}[1]{}  

\newcommand{\cA}{\mathcal{A}} \newcommand{\cB}{\mathcal{B}}
\newcommand{\cC}{\mathcal{C}} 
\newcommand{\cE}{\mathcal{E}} 
 \newcommand{\cH}{\mathcal{H}}

\newcommand{\cM}{\mathcal{M}}

 \newcommand{\cZ}{\mathcal{Z}}

\DeclarePairedDelimiter{\abs}{\lvert}{\rvert}
\DeclarePairedDelimiter{\set}{\{}{\}}
\DeclarePairedDelimiter{\parens}{(}{)}
\DeclarePairedDelimiter{\bracks}{[}{]}
\DeclarePairedDelimiter{\floor}{\lfloor}{\rfloor}
\DeclarePairedDelimiter{\ceil}{\lceil}{\rceil}

\theoremstyle{plain}
\newtheorem{theorem}{Theorem}[section]
\newtheorem{lemma}[theorem]{Lemma}

\newtheorem{proposition}[theorem]{Proposition}

\theoremstyle{definition}
\newtheorem{definition}[theorem]{Definition}

\makeatletter

\newcommand*\wthelper[2]{%
        \hbox{\dimen@\accentfontxheight#1%
                \accentfontxheight#11.1\dimen@
                $\m@th#1\widetilde{#2}$%
                \accentfontxheight#1\dimen@
        }%
}
\newcommand*\accentfontxheight[1]{%
        \fontdimen5\ifx#1\displaystyle
                \textfont
        \else\ifx#1\textstyle
                \textfont
        \else\ifx#1\scriptstyle
                \scriptfont
        \else
                \scriptscriptfont
        \fi\fi\fi3
}
\makeatother

\begin{document}

\title{Analyzing Boltzmann Samplers for Bose--Einstein Condensates with
Dirichlet Generating Functions}

\author{
Megan Bernstein\thanks{
  School of Mathematics, Georgia Institute of Technology.
  Email: \href{mailto:bernstein@math.gatech.edu}{\texttt{bernstein@math.gatech.edu}}.
	Supported in part by  National Science Foundation grant DMS-1344199.
}
\and Matthew Fahrbach\thanks{
  School of Computer Science, Georgia Institute of Technology.
  Email: \href{mailto:matthew.fahrbach@gatech.edu}{\texttt{matthew.fahrbach@gatech.edu}}.
  Supported in part by a National Science Foundation Graduate Research Fellowship under grant DGE-1650044.
}
\and Dana Randall\thanks{
  School of Computer Science, Georgia Institute of Technology.
  Email: \href{mailto:randall@cc.gatech.edu}{\texttt{randall@cc.gatech.edu}}.
  Supported in part by  National Science Foundation grants CCF-1526900, CCF-1637031, and CCF-1733812.
}
}
\date{\today}
\maketitle

\begin{abstract}
Boltzmann sampling is commonly used to uniformly sample objects of a particular
size from large combinatorial sets. 
For this
technique to be effective, one needs to prove that (1) the sampling procedure
is efficient and (2) objects of the desired size are generated with
sufficiently high probability.
We use this approach to give a provably efficient sampling algorithm for a
class of weighted integer partitions related to Bose--Einstein condensation
from statistical physics.
Our sampling algorithm is a probabilistic interpretation of the ordinary
generating function for these objects,
derived from the symbolic method of analytic
combinatorics.
Using the Khintchine--Meinardus probabilistic method to bound the rejection
rate of our Boltzmann sampler through singularity analysis of 
Dirichlet generating functions, we offer an alternative approach to analyze
Boltzmann samplers for objects with multiplicative structure.
\end{abstract}

\pagenumbering{gobble}

\newpage

\pagenumbering{arabic}
\section{Introduction}\label{sec:intro}
Bose--Einstein condensation occurs when subatomic particles known as
bosons are cooled to nearly absolute zero. These particles behave as
waves due to their quantum nature and as their temperature decreases,
their wavelength increases.
At low enough temperature, the size of the waves exceeds
the average distance between two particles and a constant fraction of
bosons enter their ground state. These particles then coalesce into a single
collective quantum wave called a \emph{Bose--Einstein condensate} (BEC).
Incredibly, this quantum phenomenon can be observed at the macroscopic scale.

In statistical physics, such thermodynamic systems are often modeled
in one of three settings:
the \emph{microcanonical ensemble}, where both the number of particles
and the total energy in the system are kept constant,
the \emph{canonical ensemble}, where the number of particles is 
constant but the energy is allowed to vary,
and the \emph{grand canonical ensemble}, where both the number of particles
and energy can vary.
%
%
A substantial amount of research has focused on  understanding
the asymptotic behavior of BECs in the microcanonical and canonical ensembles~\cite{BP83,CD14}.
Here, we give a provably efficient algorithm for uniformly sampling BEC
configurations from the {\it microcanonical ensemble} in the {\it low-temperature regime}
when the number of particles is at least the total energy in the system.
This allows us to explore thermodynamic properties of systems with thousands
of non-interacting bosons instead of relying solely on its limiting behavior.

Random sampling is widely used across scientific disciplines when exact solutions are
unavailable.  In many settings, \emph{Boltzmann samplers} have proven
particularly useful for sampling combinatorial objects of a fixed size.  The
state space $\cC$ includes configurations of all sizes, and the {\it Boltzmann
distribution} assigns a configuration $\gamma \in \cC$ probability
$\P_{\lambda}(\gamma) = \lambda^k/Z$, where~$k$ is the size of $\gamma$, $Z$ is
the normalizing constant, and $\lambda \in \R_{> 0}$ is a parameter of the
system that biases the distribution toward configurations of the desired size.
Boltzmann sampling is effective if the sampling procedure is efficient on
$(\cC, \P)$ and \emph{rejection sampling} (i.e., outputting objects of the
desired size and rejecting all others) succeeds with high enough probability to
produce samples of the correct size in expected polynomial time.

A useful example to demonstrate the effectiveness of Boltzmann sampling is {\it
integer partitions}, which arise across many areas of mathematics and
physics~(see, e.g., \cite{ABCN15,James06,Oko02}) and  turn out to be closely
related to BECs.  An \emph{integer partition} of a nonnegative integer~$n$ is a
nonincreasing sequence of positive integers that sums to $n$.
The simplest method for uniformly generating random partitions
of~$n$ relies on exact counting.
Nijenhuis and Wilf~\cite{NW78}
gave a dynamic programming algorithm for
enumerating partitions of $n$, which naturally extends to a sampling
algorithm that requires $O(n^{2.5})$ time and space. However, no suitable
closed-form expression for the number of integer partitions of $n$ is known,
thus limiting the usefulness of this approach.
Alternatively, we can use Boltzmann sampling to generate integer partitions
that are biased to have size close to $n$, augmented with rejection sampling to
output only partitions of the desired size.
Arratia and Tavar{\'e}~\cite{AT94} showed that integer partitions and many
other objects with multiplicative generating functions can be sampled from Boltzmann
distributions using independent random processes.  
Duchon~et~al.~\cite{DFLS04} generalized their approach to a systematic Boltzmann
sampling framework 
using ideas from
analytic combinatorics, 
yielding a Boltzmann sampler for
integer partitions of varying size with time and space complexity linear in the size of
partition produced.
They suggest 
tuning $\lambda$ so that the expected size of the generated object is~$n$,
which empirically gives a useful rejection rate.
Leveraging additional symmetries of partitions, Arratia and DeSalvo~\cite{AD16}
gave a sampling algorithm that runs in expected $O(\sqrt{n})$ time and space.
Taking a completely different approach, Bhakta~et~al.~\cite{BCFR17} recently
gave the first rigorous Markov chain for sampling partitions, again utilizing
Boltzmann sampling to generate samples of a desired size.

While Boltzmann sampling has been shown to be quite effective on a
vast collection of problems in statistical physics and combinatorics, several
applications lack a proof that samples will be generated from close to the
Boltzmann distribution, and even more lack rigorous arguments showing that
rejection sampling will be efficient in expectation.   
In this paper we give an provably efficient Boltzmann sampler for 
BECs and rigorously bound its rejection rate through singularity
analysis of Dirichlet generating functions.
Our techniques naturally extend to many families of \emph{weighted partitions}, which
generalize integer partitions and BECs.
\subsection{Bose--Einstein Condensates}
We study Bose--Einstein condensation in an idealized microcanonical setting
with limited interactions between particles.
In the physics community this represents
configurations of zero-spin particles in a $d$-dimensional isotropic 
harmonic trap in the 
absence of two-body interactions and particle-photon interaction.
We focus on the case $d=3$, but the results generalize to higher dimensions. 
Combinatorially, we interpret BECs as weighted partitions
with $b_k = \binom{k+3-1}{3-1}$ types of summands of size $k$.
The $b_k$ summands correspond to the degenerate energy states of a boson with energy $k$.
We view such energy states as multisets of three different colors with cardinality~$k$.
BEC configurations are unordered collections of bosons and can be understood as 
multisets of bosons, or equivalently weighted partitions.
We represent BECs graphically by coloring Young diagrams.
Each column corresponds to the energy state of a particle,
and the columns are sorted lexicographically to give a partition of particles.
In the microcanonical ensemble with $m$ particles and energy $n$, the number of
particles in their ground state is $m$ minus the width of the Young diagram.
Bose--Einstein condensation occurs when the width of the average Young diagram
is at most a constant fraction of~$m$.

In the language of analytic combinatorics, BECs are the combinatorial class
$\MSET(\MSET_{\ge 1}(3\cZ))$.
For example, if $n=2$ there are 12 possible configurations.
When there is one particle with energy~2 we have
$\{\{1,1\}\}$, $\{\{1,2\}\}$, $\{\{1,3\}\}$, $\{\{2,2\}\}$, $\{\{2,3\}\}$, $\{\{3,3\}\}$,
and when there are two particles each with energy~1 we have
$\{\{1\},\{1\}\}$, $\{\{1\},\{2\}\}$, $\{\{1\},\{3\}\}$, $\{\{2\},\{2\}\}$
$\{\{2\},\{3\}\}$, $\{\{3\},\{3\}\}$.
If~$n=3$ there are 38 possible configurations:
10 if there is one particle with energy~3,
18 that emerge for the Young diagram with shape $(2,1)$
when one particle has energy~2 and one has energy~1 (\Cref{fig:weighted-partitions-demonstration}),
and 10 when there are three particles with energy~1.

\begin{figure*}
  \centering
  \includegraphics[width=0.98\textwidth]{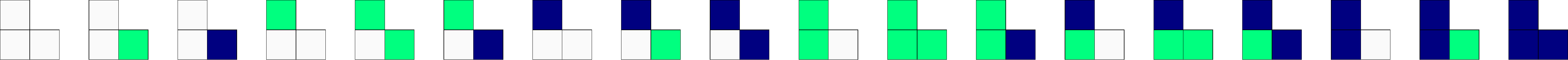}
  \caption{Young diagrams of Bose--Einstein condensates with shape $(2,1)$, where
           the colors (gray, green, blue) correspond to the numbers $(1,2,3)$.
           }
  \label{fig:weighted-partitions-demonstration}
\end{figure*}


\subsection{Results}
Our complexity analysis follows the conventions in~\cite{BFP10,FFP07} and assumes
a real-arithmetic model of computation, an oracle that evaluates a generating
function within its radius of convergence in constant time, and a root-finding
oracle.  We give a new method for rigorously analyzing algorithms that sample from the
Bose--Einstein distribution in low-temperature microcanonical ensembles
when the number of particles exceeds the total energy.\footnote{
  In this setting each weighted partition of size $n$ uniquely corresponds
  to a BEC configuration in the microcanonical ensemble.}
In particular, we give a
provably efficient algorithm for uniformly sampling BECs
by constructing a linear-time Boltzmann sampler using the framework
established in~\cite{FFP07}, and then bounding its rejection rate through
singularity analysis of an associated Dirichlet generating function.

\begin{restatable}{theorem}{newmaintheorem}
\label{thm:new-main-theorem}
There exists a uniform sampling algorithm for Bose--Einstein condensates of size~$n$
that runs in expected $O(n^{1.625})$ time and uses $O(n)$ space.
\end{restatable}

\noindent
This algorithm generates samples of size $n$ exactly from the uniform
distribution in expected polynomial time.  This
allows us, for example,  to rigorously study the expected width of Young diagrams
arising from random configurations 
(or, equivalently, the fraction of particles in a BEC in their ground state)
without relying on the limiting
properties given in~\cite{Yak2012}.

The singularity analysis used in the proofs generalizes to a broader family of
weighted partitions, including integer partitions and plane
partitions~\cite{BFP10}.
Let a \emph{positive integer sequence of degree~$r$} be a sequence of positive
integers $(b_k)_{k=1}^\infty$ such that $b_k = p(k)$ for some polynomial
\[
  p(x) = a_0 + a_1 x + \dots + a_r x^r \in \R[x],
\]
with $\deg(p) = r$.
We show how the rightmost pole of the Dirichlet generating function for
$(b_k)_{k=1}^\infty$ and its residue are related to $a_r$, the leading
coefficient of $p(x)$, which we then use to establish rigorous rejection rates
for Boltzmann sampling.

\begin{restatable}{theorem}{generaltheorem}
\label{thm:general-theorem}
There exists a uniform sampling algorithm for any class of weighted partitions
parameterized by a positive integer sequence of degree $r$ for
objects of size~$n$ that runs in expected time $O(n^{r+1 + (r+3)/(2r+4)})$ and
uses $O(n)$ space.
\end{restatable}

For a fixed degree $r$, the number of samples needed in expectation is
$O(n^{(r+3)/(2r+4)})$, which is asymptotically tight 
by Theorem~\ref{thm:local-limit}.
In particular, when $r=0$ (as in the case of integer partitions)
we need $O(n^{3/4})$ samples,
and as $r \rightarrow \infty$ the number of required samples converges to~$O(\sqrt{n})$.
Independently, DeSalvo and Menz~\cite{DM16} recently developed a new probabilistic model
that gives a central limit theorem for the same family of weighted
partitions and circumvents singularity analysis
of Dirichlet generating functions.
\subsection{Techniques}
Probabilistic interpretations of ordinary generating functions for weighted
partitions are useful for producing samples from Boltzmann distributions in
polynomial time~\cite{AT94,FFP07}, but rejection sampling is not always
guaranteed to be efficient.  We use the
\emph{Khintchine--Meinardus probabilistic method}~\cite{Gra16, GS12, GSE08}
to establish that rejection sampling is efficient
for a broad class of weighted partitions that includes BECs.
We also use the Boltzmann sampling framework based on the
symbolic method from analytic combinatorics
to give a improved algorithms for BECs, instead of simply using
a geometric random variable for each degenerate energy state
(Theorem~\ref{thm:new-main-theorem} vs. Theorem~\ref{thm:general-theorem} with~$r=2$).

The goal of the Khintchine--Meinardus probabilistic is to asymptotically
enumerate combinatorial objects through singularity analysis of Dirichlet
generating functions. For algorithmic purposes, it is useful for bounding
rejection rates.
Only recently were such techniques extended to
handle Dirichlet series with multiple poles on the real axis~\cite{GS12},
which is a necessary advancement for analyzing BECs and other
weighted partitions parameterized by non-constant integer sequences.
In this paper, we show that the Dirichlet series for classes of
weighted partitions parameterized by positive integer sequences are linear
combinations of shifted Riemann zeta functions, and thus amenable to singularity
analysis.
Using bounds for the Riemann zeta function from analytic number theory,
we show that the local limit theorem in~\cite{GS12} holds for BECs.
We also view the parameterizing polynomial $p(x)$ as a Newton-interpolating
polynomial to lower bound the residue of the rightmost pole of the Dirichlet
generating functions, which leads to a more convenient lemma for rejection
rates.

Additionally,
we develop a tail inequality
for the negative binomial distribution (\Cref{lem:d-root-convolution})
that is empirically much tighter than a Chernoff-type inequality 
when used in the analysis of a subroutine of the main sampling algorithm.
We also believe the
singularity analysis of Dirichlet series in this paper will be valuable for a
wide variety of sampling problems in computer science
when the objects of interest can be decomposed into non-interacting components
and when transfer theorems for ordinary generating functions do not work.



\section{Preliminaries}
We start by introducing fundamental ideas about Boltzmann samplers for
combinatorial classes. 
Then we use the symbolic method to define Bose--Einstein condensates
and present a local limit theorem from the
Khintchine--Meinardus probabilistic method for weighted partitions.
These are the main components of our analysis. 

\subsection{Boltzmann Sampling}
A combinatorial class $\cC$ is a finite or countably infinite set 
equipped with a size function $|\cdot| : \cC \rightarrow \Z_{\ge 0}$
such that 
the number of elements of any given size is finite.
For a given class $\cC$, let $c_n$ be the number of elements of size $n$.
The counting sequence of $\cC$ is the integer sequence
$(c_n)_{n=0}^\infty$, and the ordinary generating function of $\cC$ is
\[
  C(z) \DEF \sum_{n=0}^\infty c_n z^n \DEF \sum_{\gamma \in \cC} z^{|\gamma|}.
\]

\begin{definition}
The \emph{Boltzmann distribution} of a class $\cC$ 
parameterized by $\lambda \in (0, \rho_C)$
is the probability distribution, for all $\gamma \in \cC$, defined as
\begin{equation*}
  \P_{\lambda}(\gamma) \DEF \frac{\lambda^{|\gamma|}}{C(\lambda)},
\end{equation*}
where $\rho_C$ is the radius of convergence of $C(z)$.
\end{definition}

\begin{definition}
A \emph{Boltzmann sampler} $\Gamma C(\lambda)$ 
is an algorithm that generates objects from a class~$\cC$ according the
Boltzmann distribution with parameter $\lambda$.
\end{definition}

The size of an object generated by $\Gamma C(\lambda)$ is a random variable
denoted by $U$ with the probability distribution
\[
  \P_\lambda(U = n) = \frac{c_n \lambda^{n}}{C(\lambda)}.
\]
All objects of size $n$ occur with equal probability,
so if $\Gamma C(\lambda)$ returns an object $\gamma$ of size $n$, then $\gamma$
is a uniform random sample among all size $n$ objects in $\cC$.
Therefore, we can use \emph{rejection sampling} to 
generate objects of size $n$ uniformly at random.
For this technique to be effective, we need both
an efficient sampling algorithm over the pair $(\cC, \P_\lambda)$
and a provably low rate of rejection.

Assuming that $\cC$ is infinite and $c_0 \ne 0$,
we maximize the probability of generating an object of size~$n \ge 1$ 
by tuning the Boltzmann sampler so that $\E_{\lambda}[U] = n$ and
denote this solution by $\lambda_n$.
To see why this strategy works, observe that
\begin{equation}
\label{eqn:max-prob}
  \frac{\mathrm{d}}{\mathrm{d}\lambda} \P_{\lambda}\parens*{U = n}
  = \frac{c_n \lambda^{n-1}}{C(\lambda)} \parens*{n - \E_{\lambda}[U]},
\end{equation}
where $\E_\lambda[U] = \lambda C'(\lambda) / C(\lambda)$.
Because $\cC$ contains objects of varying size,
we have
\begin{equation}
\label{eqn:strictly-increasing}
  \frac{\mathrm{d}}{\mathrm{d} \lambda} \E_\lambda[U] = \frac{\Var_\lambda(U)}{\lambda} > 0.
\end{equation}
Thus $\E_{\lambda}[U]$ is strictly increasing,
so $\lambda_n$ is unique.
Together with \Cref{eqn:max-prob},
this implies that  $\P_\lambda(U=n)$ is maximized at $\lambda_n$.
The equality in \Cref{eqn:strictly-increasing} is a property of Boltzmann
distributions~\cite{DFLS04},
and the inequality is true
because $\lambda \in (0, \rho_C)$ and $U$ is a nonconstant random variable.

\subsection{Symbolic Method}
We use the symbolic method of analytic combinatorics to construct
the class of Bose--Einstein condensates
and then utilize the Boltzmann sampling framework developed in~\cite{DFLS04,Fla07,FFP07}.
The primitive combinatorial classes in the symbolic method 
are the neutral class $\cE$ and
the atomic class~$\cZ$.
The class~$\cE$ contains a single element of
size~$0$ called a neutral object,
and the class~$\cZ$ contains a single element of size~$1$ called an atom.
Neutral objects are used to mark objects as different, and
atoms are combined to form combinatorial objects.
We can express a rich family of discrete structures
using these primitive classes with the following operators.

\begin{definition}
The \emph{combinatorial sum} of $\cA$ and $\cB$ is
\[
  \cC = \cA + \cB \DEF \parens*{\cE_1 \times \cA}
  \cup \parens*{\cE_2 \times \cB},
\]
where $\cE_1$ and $\cE_2$ are different neutral classes.
The size of an element in $\cC$ is the same as in its class of origin,
and the ordinary generating function for $\cC$ is $C(z) = A(z) + B(z)$.
\end{definition}

\begin{definition}
The \emph{Cartesian product} of $\cA$ and $\cB$ is
\[
  {\cC = \cA \times \cB \DEF \{(\alpha,\beta) : \alpha \in \cA, \beta \in \cB\}}.
\]
The size of the pair $\gamma = (\alpha,\beta) \in \cC$ is defined as
$|\gamma|_{\cC} \DEF |\alpha|_{\cA} + |\beta|_{\cB}$, and
the generating function for~$\cC$ is $C(z) = A(z)B(z)$.
\end{definition}

\begin{definition}
The \emph{sequence operator} of a class $\cB$ with $b_0 = 0$
is the infinite sum
\[
  \cC = \SEQ(\cB) \DEF \cE + \cB + (\cB \times \cB) +
  (\cB \times \cB \times \cB) + \cdots,
\]
and the generating function for $\cC$ is
\begin{equation*}
  C(z) = 1 + B(z) + B(z)^2 + B(z)^3 + \dots = \frac{1}{1-B(z)}.
\end{equation*}
\end{definition}

\begin{definition}
The \emph{multiset operator} of a class $\cB$ with $b_0 = 0$ is
\[
  \cC = \MSET(\cB) \DEF \prod_{\beta \in \cB} \SEQ(\{\beta\}),
\]
and the generating function for $\cC$ is
\begin{align}
\label{eqn:ac-multiset-gf-def}
  C(z) &= \prod_{\beta \in \cB} \parens*{1-z^{|\beta|}}^{-1}\\ \nonumber
  &= \prod_{k=1}^\infty \parens*{1 - z^k}^{-b_k}\\
  &= \prod_{k=1}^\infty \exp\parens*{ \frac{1}{k} B\parens*{z^k}}. \nonumber
\end{align}
The final equality is an exp-log transform~\cite{FS09}.
\end{definition}

Multisets are essential to our analysis because they demonstrate how
Bose--Einstein condensates decompose into combinatorial atoms.
Bose--Einstein condensates and integer partitions are generalized by weighted
partitions, which are the central objects in the Khintchine--Meinardus
probabilistic method.

\begin{definition}
The class $\cC$ of \emph{weighted partitions} 
with $b_k$ different types of summands of size~$k \ge 1$ is implicitly defined by
the generating function
\[
C(z) \DEF \sum_{n=0}^\infty c_n z^n
  \DEF \prod_{k=1}^\infty \parens*{1-z^k}^{-b_k}.
\]
Equivalently,
$\cC = \MSET(\cB)$ is parameterized by the class $\cB$ of permissible summands. 
\end{definition}

From here
onward, we use $\cC$ to denote classes of weighted partitions and
$\cB$ to denote the corresponding class of summands.
For uniform sampling, it is beneficial to work with the truncated
class of weighted partitions~$\cC_{n}$ with the generating function
\[
  C_{n}(z) \DEF \prod_{k=1}^n \parens*{1 - z^k}^{-b_k},
\]
since it completely contains the target set of objects of size $n$.  
Analogously, we define a random variable for the size of an object
produced by a Boltzmann sampler for the truncated class $\cC_n$.

\begin{definition}\label{def:U}
Let $U_n$ denote the random variable for the size of an object
generated by $\Gamma C_{n}(\lambda)$.
\end{definition}

Equipped the notions of weighted partitions and the symbolic method, we can easily
construct Bose--Einstein condensates in a way that is amenable to efficient 
uniform sampling.

\begin{definition}
\emph{Bose--Einstein condensates} are weighted partitions 
with the parameters
\[
  b_k = \binom{k+2}{2}.
\]
In the language of analytic combinatorics 
they are the combinatorial class $\MSET(\MSET_{\ge 1}(3\cZ))$.
\end{definition}

\noindent
The parameterizing class $\MSET_{\ge 1}(3\cZ)$ is the set of
all nonempty multisets of $3$ different colored atoms.
There are $\binom{k+3-1}{3-1}$ such multisets of size $k$, each corresponding
to a different summand.
From a physics point of view,
multisets of size $k$ in $\MSET_{\ge 1}(3\cZ)$ are isomorphic to the
three-dimensional degenerage energy states of a boson with energy $k$.
A Bose--Einstien condensate is an unordered collection of bosons,
so an object of size~$n$ in $\MSET(\MSET_{\ge 1}(3\cZ))$ uniquely corresponds to
a Bose--Einstein condensate with total energy $n$.

\subsection{\hspace{-0.30cm}Khintchine--Meinardus \hspace{-0.20cm}Probabilistic \hspace{-0.20cm}Method}

Under somewhat restrictive conditions,
Meinardus~\cite{Mei54} established an asymptotic equivalence between
the number of weighted partitions $c_n$
and the analytic behavior of the Dirichlet series 
\begin{equation*}
  D(s) \DEF \sum_{k=1}^\infty b_k k^{-s},
\end{equation*}
using the saddle-point method~\cite{FS09},
where $s = \sigma + it$ is a complex variable.
Granovsky, Stark, and Erlihson~\cite{GSE08}
extended Meinardus' theorem to new multiplicative combinatorial objects
using Khintchine's probabilistic method~\cite{Khi98}.
Granovsky and Stark \cite{GS12} later generalized their 
results to weighted partitions such that
$D(s)$ has multiple singularities on the positive real axis,
which includes the class of Bose--Einstein condensates.
To use the Khintchine--Meinardus probabilistic method for 
weighted partitions, we must show that the series $D(s)$ satisfies the following
conditions:


\begin{enumerate}[(I)]
\item
The Dirichlet series $D(s)$ has $r \ge 1$
simple poles at real positions $0 < \rho_1 < \rho_2 < \dots < \rho_r$ with
positive residues $A_1, A_2, \dots, A_r$, and it
is
analytic in the half-plane $\sigma > \rho_r > 0$.
Moreover, there is a constant $0 < C_0 \le 1$ such that the function
$D(s)$ has a meromorphic continuation to the half-plane
\[
  \cH = \set*{s : \sigma \ge -C_0},
\]
on which it is analytic except for above the $r$ simple poles.

\item
There is a constant $C_1 > 0$ such that
\[
  D(s) = O\parens*{|t|^{C_1}},
\]
uniformly for $s = \sigma + it \in \cH$, as $t \rightarrow \infty$.

\item
For $\delta > 0$ small enough and some $\varepsilon > 0$,
\[
  2\sum_{k=1}^\infty b_k e^{-k\delta} \sin^2(\pi k \alpha)
    \ge \cM \parens*{1 + \frac{\rho_r}{2} + \varepsilon}
  \abs{\log \delta},
\]
for all $\sqrt{\delta} \le \abs{\alpha} \le 1/2$
and $\cM = 4/\log 5$.
\end{enumerate}

The local limit theorem that follows is the crux of the proof for \cite[Theorem~1]{GS12}.
We use it to prove asymptotically tight rejection rates for our Boltzmann
sampler as a function of the rightmost pole of $D(s)$.
In the statement,
$\Gamma(z)$ is the gamma function and $\zeta(s)$ is the Riemann zeta function.

\begin{theorem}[{\cite{GS12}}]\label{thm:local-limit}
If conditions (I)--(III) hold, 
\begin{align*}
  \P_{\lambda_n}\parens*{U_n = n} &\sim \frac{1}{\sqrt{2\pi\Var\parens*{U_n}}}\\
                                  &\sim \frac{1}{\sqrt{2\pi K_2}} \parens*{\frac{K_2}{\rho_r+1}}^{\frac{2+\rho_r}{2(\rho_r+1)}}
  n^{-\frac{2+\rho_r}{2(\rho_r + 1)}},
\end{align*}
as $n \rightarrow \infty$, where $K_2 = A_r \Gamma(\rho_r + 2)\zeta(\rho_r + 1)$ is a
  constant.
\end{theorem}

\section{Sampling Bose--Einstein Condensates}
\label{sec:new-main-result}

We can now present our main algorithm for uniformly sampling Bose--Einstein
condensates.  We bound the rejection rate of its underlying Boltzmann sampler
using singularity analysis of a Dirichlet generating function
and then analyze the complexity of the overall algorithm to prove
Theorem~\ref{thm:new-main-theorem}.



\begin{algorithm}[H]
  \caption{Uniform sampling algorithm for BECs.} 
    \label{alg:bec-algorithm}
    \begin{algorithmic}[1]
      \Procedure{RandomBEC}{$n$}
        \State $\lambda_n \gets \text{Solve} \sum_{k=1}^n k \binom{k+2}{2} \lambda^k/(1-\lambda^k) = n$
        \Repeat 
          \State $\gamma \gets \Gamma \MSET[\MSET_{1..n}(3\cZ)](\lambda_n)$
        \Until{$|\gamma| = n$}
        \State \textbf{return} $\gamma$
      \EndProcedure
    \end{algorithmic}
\end{algorithm}

Before discussing the Boltzmann samplers in \Cref{alg:bec-algorithm},
we define several probability distributions that are fundamental to the
following subroutines.

\begin{definition}
Let $\GEO(\lambda)$ denote the \emph{geometric distribution} with success
probability~$\lambda$ and probability density function
\[
  \P_\lambda(k)=(1-\lambda)^k \lambda,
\]
for all $k \in \Z_{\ge 0}$.
\end{definition}
\begin{definition}
Let $\POIS(\lambda)$ denote the \emph{Poisson distribution} with
rate parameter $\lambda$ and probability density function
\[
  \P_\lambda(k) = \frac{\lambda^k}{e^\lambda k!},
\]
for all $k \in \Z_{\ge 0}$.
The zero-truncated Poisson distribution
$\POIS_{\ge 1}(\lambda)$ with rate parameter $\lambda$
has the probability density function
\[
  \P_\lambda(k) = \frac{\lambda^k}{(e^\lambda - 1)k!}.
\]
\end{definition}
\begin{definition}
Let $\NB(r,\lambda)$ denote the \emph{negative binomial distribution} with
$r$ failures, success probability $\lambda$, and probability density function
\[
  \P_\lambda(k) = \binom{k+r-1}{r-1} \lambda^k (1-\lambda)^r,
\]
for all $k \in \Z_{\ge 0}$.
The zero-truncated negative binomial distribution is
$\NB_{\ge 1}(r, \lambda)$.
\end{definition}

One of the two main subroutines in \Cref{alg:bec-algorithm} is a
template Boltzmann sampler for the class $\MSET(\cA)$, where
$\cA$ is any combinatorial class with $a_0 = 0$.
This algorithm repeatedly calls the Boltzmann sampler of the input class $\Gamma A(\lambda^k)$
(for various values of $k$) and is part of the Boltzmann sampling
framework for combinatorial classes that can be constructed with the symbolic method~\cite{FFP07}.

\begin{algorithm}[H]
  \caption{Boltzmann sampler for $\MSET(\cA)$.}
    \label{alg:multiset-sampler}
    \begin{algorithmic}[1]
      \Procedure{$\Gamma \MSET[\cA]$}{$\lambda$}
        \State $\gamma \gets \text{Empty associative array}$
        \State $k_0 \gets \textsc{MaxIndex}(A, \lambda)$
        \For{$k=1$ \textbf{to} $k_0$}
          \If{$k < k_0$}
            \State $m \gets \POIS(A(\lambda^k)/k)$
          \Else
            \State $m \gets \POIS_{\ge 1}(A(\lambda^k)/k)$
          \EndIf
          \For{$j=1$ \textbf{to} $m$}
            \State $\alpha \gets \Gamma A(\lambda^k)$ 
            \State $\gamma[\alpha] \gets \gamma[\alpha] + k$
          \EndFor
        \EndFor
        \State \textbf{return} $\gamma$
      \EndProcedure
    \end{algorithmic}
\end{algorithm}

In particular,
\Cref{alg:multiset-sampler} is a manifestation of the generating function
for weighted partitions.
It captures the exp-log transform in~\Cref{eqn:ac-multiset-gf-def}
by using the property that a geometric random variable can be
decomposed into an infinite sum of independent, scaled Poisson random variables.
We direct the reader to the proof of \cite[Proposition 2.1]{FFP07} for more details.
The function $\textsc{MaxIndex}(A,\lambda)$ samples from the 
distribution with cumulative density function
\begin{align}
\label{eqn:inversion-method}
  \P_\lambda(U \le k) &= \frac{\prod_{j = 1}^{k} \exp\parens*{\frac{1}{j} A\parens*{\lambda^j}}}{
    \prod_{j = 1}^{\infty} \exp\parens*{\frac{1}{j} A\parens*{\lambda^j}}}\\
  &= \prod_{j = k+1}^\infty
    \exp\parens*{-\frac{1}{j} A\parens*{\lambda^j}}, \nonumber
\end{align}
for all integers $k \ge 1$. 
Note that we give the corrected expression for \Cref{eqn:inversion-method} that
originally appeared in~\cite{BFP10}.

\begin{proposition}[{\cite{FFP07}}]
\label{prop:mset-sampler-correctness}
\Cref{alg:multiset-sampler}
is a valid Boltzmann sampler for $\MSET(\cA)$.
Moreover, if the time and space complexities of $\Gamma A(\lambda)$ are, in the worst
case, linear in the size of the object produced, the time and space
complexities of \Cref{alg:multiset-sampler} are also linear in the size of the
object produced.
\end{proposition}

The second subroutine in \Cref{alg:bec-algorithm} is a Boltzmann sampler
for the nonempty multisets of~$d$ different colored atoms with size at
most $n$.
This parameterizes the truncated class of Bose--Einstein condensates when $d=3$,
and it is necessary for the Khintchine--Meinardus probabilistic method.
It suffices to use the Boltzmann sampler $\Gamma \MSET_{\ge 1}(d\cZ)$
and rejection sampling by \Cref{lem:middle-rejection}.

\begin{algorithm}[H]
  \caption{Boltzmann sampler for multisets.}
    \label{alg:bec-subroutine}
    \begin{algorithmic}[1]
      \Procedure{$\Gamma \MSET_{1..n}(d\cZ)$}{$\lambda$}
        \Repeat  
          \State $m \gets \NB_{\ge 1}(d, \lambda)$
        \Until{$m \le n$}
        \State \textbf{return} $\textsc{RandomMultiset}(m, d)$
      \EndProcedure
    \end{algorithmic}
\end{algorithm}

\noindent
The crucial observation in \Cref{alg:bec-subroutine} is that
a negative binomial experiment with
$d$ failures can be interpreted as a multiset of $d$ different colored atoms using 
the classical combinatorial idea of stars and bars.
For a given experiment,
a successful trial adds an atom of the current color and
a failure is a bar that separates atoms of different colors.
Once $m$ is determined,
the function $\textsc{RandomMultiset}(m,d)$
returns one of the $\binom{m+d-1}{d-1}$ multisets of size $m$ uniformly at random.

\subsection{Tuning the Boltzmann Sampler}
\label{sec:tuning}



The factorization of the generating function for weighted partitions has a
useful probabilistic interpretation in the context of Boltzmann sampling.
We can iterate over all types of summands of size at most $n$ and use
independent geometric random variables to determine how many parts each
type contributes to the final object~\cite{AT94, FFP07}.
It follows that the random variable for the size of the object
drawn from the Boltzmann sampler is
\[
  U_{n} = \sum_{k=1}^n \sum_{j=1}^{b_k} kY_{k,j},
\]
where
$Y_{k,j} \sim \GEO\parens{1-\lambda^k}$.

\begin{lemma}\label{lem:expectation2}
We have
\[
  \E_{\lambda}\bracks*{U_n} = \sum_{k=1}^n k b_k \parens*{\frac{\lambda^k}{1 - \lambda^k}}.
\]
\end{lemma}
\begin{proof}
The result follows from the linearity of expectation
and the mean of the variables $Y_{k,j}$.
\end{proof}

Since $\E_{\lambda}[U_n]$ is strictly increasing
\Cref{eqn:strictly-increasing}, we can use \Cref{lem:expectation2} and a
root-finding algorithm such as the bisection method to compute an
$\varepsilon$-approximation of $\lambda_n$ in time $O(n\log (\varepsilon^{-1}))$.
In this paper, however, we assume an oracle that returns the exact
value of~$\lambda_n$ in constant time and defer the numerical analysis needed
for the accuracy of the $\varepsilon$-approximation. 
Since we are dealing with analytic functions, a quantitative statement
about the continuity of the local limit theorem near~$\lambda_n$ should
be achievable.

\subsection{Rejection Rate of \Cref{alg:bec-algorithm}}
\label{sec:rejection-rates-1}
We generalize the analysis in this subsection from Bose--Einstein condensates
to the class of weighted partitions parameterized by a positive integer
sequence of degree $r$.
We show that Theorem~\ref{thm:local-limit} holds for all such weighted partitions
and bound the residue of the rightmost pole of the series $D(s)$
to give rejection rates for these Boltzmann samplers as a function of the degree~$r$.

\begin{definition}
Let a \emph{positive integer sequence of degree $r$} be a sequence of positive
integers $(b_k)_{k=1}^\infty$ such that $b_k = p(k)$ for some polynomial
\[
  p(x) = a_0 + a_1 x + \dots + a_r x^r \in \R[x],
\]
with $\deg(p) = r$.
\end{definition}

To interface with the Khintchine--Meinardus probabilistic method, 
we apply results from analytic number theory about the Riemann zeta function.
The Riemann zeta function $\zeta(s)$, $s = \sigma + it$,
is defined as the analytic continuation of the convergent series
\[
  \zeta(s) = \sum_{n=1}^\infty \frac{1}{n^s},
\]
for all $\sigma > 1$,
into the entire complex plane. The only singularity of $\zeta(s)$ is a simple pole at $s=1$.

The following lemma shows that the Dirichlet generating function $D(s)$
for $(b_k)_{k=1}^\infty$
is a linear combination of Riemann zeta functions that are translated and
scaled by the coefficients of $p(x)$.
With this, we can easily compute the residues of the poles of $D(s)$
and satisfy conditions (I)--(III).

\begin{restatable}{lemma}{mainlemma}
\label{lem:main-lemma}
If $(b_k)_{k = 1}^{\infty}$ is a positive integer sequence of degree $r$,
its Dirichlet generating series is
\begin{equation*}
  D(s) = \sum_{k=0}^{r} a_k \zeta(s - k),
\end{equation*}
which satisfies conditions (I)--(III).
Moreover, $D(s)$ has at most $r+1$ simple poles
on the positive real axis at positions
$\rho_k = k+1$ with residue $A_k = a_k$ if and only if $a_k \ne 0$,
for $k \in \{0,1,\dots,r\}$.
\end{restatable}


\begin{proof}
The Riemann zeta function converges uniformly and
is analytic on $\C\setminus\{1\}$, so
\begin{align*}
  D(s) &= \sum_{k=1}^\infty \frac{b_k}{k^s}\\
       &= \sum_{k=1}^\infty \frac{a_0 + a_1k + \dots + a_r k^r}{k^s}\\
       &= \sum_{k=0}^r a_k \zeta(s-k).
\end{align*}
As $\zeta(s)$ has a simple pole at $s=1$ with residue~$1$, the claim about the
poles of $D(s)$ follows, and condition (I) is satisfied by letting $C_0 = 1$.
For condition (II), we let $C_1 = 2 + r$ since $D(s)$ is a linear
combination of shifted zeta functions and
use \cite[Section 8.2]{BD04}:
\begin{equation*}\label{eqn:zeta-bounds}
  \zeta(s) = \begin{cases}
    O\parens*{t^{1/2-\sigma}} & \text{if } \sigma < 0,\\
    O(t) & \text{if } 0 \le \sigma \le 1,\\
    O(1) & \text{if } 1 < \sigma.
  \end{cases}
\end{equation*}

To show condition (III),
we use an approach similar to the proof of
\cite[Lemma 1]{GSE08}, which employs the following inequality of
Karatsuba and Voronin~\cite[Section 4.2, Lemma 1]{KV92}
for trigonometric sums related to the Riemann zeta function:
For all positive integers~$m$, 
\begin{equation}\label{eqn:sin-bound}
  2\sum_{k=1}^m \sin^2(\pi k \alpha)
  \ge m\parens*{1 - \min\parens*{1, \frac{1}{2m|\alpha|}}}.
\end{equation}
By assumption, $(b_k)_{k=1}^\infty$ is a sequence of positive integers and
$0 < \sqrt{\delta} \le |\alpha| \le 1/2$, so 
\begin{align*}
  2\sum_{k=1}^\infty b_k e^{-k\delta} \sin^2(\pi k \alpha)
  &\ge 2\sum_{k=1}^m e^{-k\delta} \sin^2(\pi k \alpha)\\
  &\hspace{-0.1cm}\ge e^{-m\delta} m\parens*{1 - \min\parens*{1, \frac{1}{2m|\alpha|}}}.
\end{align*}
Using \Cref{eqn:sin-bound}, we let
$m = \ceil{1/(2|\alpha|) + 1/\delta} \ge 1$
so that
\begin{align*}
  2\sum_{k=1}^\infty b_k e^{-k\delta} \sin^2(\pi k \alpha)
  &\ge e^{-m\delta} \parens*{m - \frac{1}{2|\alpha|}}\\
  &\ge e^{-\parens*{\frac{1}{2|\alpha|} + \frac{1}{\delta} + 1}\delta} \delta^{-1}\\
  &\ge e^{-\parens*{\frac{\sqrt{\delta}}{2} + 1 + \delta}} \delta^{-1}.
\end{align*}
Recall that the position of the pole $\rho_r = r+1$ is fixed and $\cM$ is
constant. It follows that
\begin{equation*}
  e^{-\parens*{\frac{\sqrt{\delta}}{2} + 1 + \delta}} \frac{1}{\delta \abs{\log \delta}}
  \ge \cM \parens*{1 + \frac{\rho_r}{2} + \varepsilon},
\end{equation*}
for $\delta$ sufficiently small and $\varepsilon = 1$.
\end{proof}

To conveniently use Theorem~\ref{thm:local-limit} in the analysis of
Algorithm~\ref{alg:bec-algorithm},
we lower bound the rightmost residue $A_r$ using \Cref{lem:main-lemma}
and the method of finite differences.
In particular, we bound the coefficients of $p(x)$ in the
binomial basis as illustrated in~\cite{BFT16}.
Define the \emph{forward difference} operator~$\Delta$ to be
\[
  \Delta p(x) = p(x+1) - p(x).
\]
Higher order differences are given by
\[
  \Delta^n p(x) \DEF \Delta^{n-1} p(x + 1) - \Delta^{n-1} p(x).
\]
Viewing $p(x)$ as a Newton interpolating polynomial, we can write
\[
  p(x) = \sum_{j=0}^r \Delta^j p(0) \binom{x}{j}.
\]

\begin{lemma}[{\cite{Sta12}}]
\label{lem:stanley}
Let $p(x) \in \R[x]$ be a polynomial of degree $r$.
We have $p(n) \in \Z$, for all $n \in \Z$, if and only if
\[
  \Delta^j p(0) \in \Z,
\]
for all $0 \le j \le r$.
\end{lemma}

\begin{restatable}{lemma}{rejectionrate}
\label{lem:rejection-rate}
Let $(b_k)_{k=1}^\infty$ be any positive integer sequence of degree $r$.
For $n$ sufficiently large,
\[
  \P_{\lambda_n}\parens*{U_n = n} \ge \frac{1}{2\sqrt{2 \pi}}
  \parens*{(r+2)n}^{-\frac{r+3}{2(r+2)}}.
\]
\end{restatable}

\begin{proof}
Theorem \ref{thm:local-limit} holds for $U_n$ by \Cref{lem:main-lemma}.
By assumption $(b_k)_{k=1}^\infty$ is a positive integer sequence of degree~$r$,
so \Cref{lem:stanley} implies that $p(x)$ has integral coefficients
$\Delta^{j}p(0)$
in the binomial basis.
Thus, $\Delta^r p(0)$ is a positive integer, which implies that
\[
  a_r \ge \frac{1}{r!}.
\]
The residue $A_r = a_r$ by \Cref{lem:main-lemma},
so we can $K_2$ by
\begin{align*}
  K_2 &= A_r \Gamma(\rho_r + 2) \zeta(\rho_r + 1)\\
  &\ge \frac{1}{r!} (r+2)!\\
  &\ge 1,
\end{align*}
since $\rho_r = r+1$ and $\zeta(n) \ge 1$, for all $n \ge 2$.
It follows that for $\varepsilon = 1/2$ and $n$ sufficiently large,
we have
\begin{align*}
  \P_{\lambda_n}\parens*{U_n = n} &\ge (1-\varepsilon) \frac{1}{\sqrt{2\pi K_2}} \parens*{\frac{K_2}{(\rho_r+1)n}}^{\frac{2+\rho_r}{2(\rho_r+1)}}\\
  & \ge \frac{1}{2\sqrt{2\pi}} K_2^{\frac{1}{2(r+2)}}
  \parens*{(r+2)n}^{-\frac{r+3}{2(r+2)}}\\
  & \ge \frac{1}{2\sqrt{2\pi}} \parens*{(r+2)n}^{-\frac{r+3}{2(r+2)}}, 
\end{align*}
as desired. 
\end{proof}

\subsection{Proving the Main Theorems}
\label{sec:main-theorem-proof}

Recall that we follow the convention of using the real-arithmetic model
of computation and
an oracle that evaluates a
generating function within its radius of convergence
in constant time~\cite{BFP10,FFP07}.
A consequence of this is that we can iteratively sample an integer~$m$ from the distributions
$\GEO(\lambda)$, $\POIS(\lambda)$, $\textsc{MaxIndex}(A,\lambda)$, etc., in
$O(m)$ time.
We restate Theorem~\ref{thm:new-main-theorem} 
and Theorem~\ref{thm:general-theorem} for convenience.

\newmaintheorem*

\begin{proof}
We analyze the complexities of \Cref{alg:bec-algorithm}.
In the tuning step, we can compute the exact
value of $\lambda_n$ in $O(1)$ time using a root-finding oracle and \Cref{lem:expectation2}.
(We can compute an $\varepsilon$-approximation in time $O(n \log(\varepsilon^{-1}))$.)
We invoke the Boltzmann sampler $\Gamma \MSET[\MSET_{1..n}(3\cZ)](\lambda_n)$
at most $O(n^{5/8})$ times in expectation by \Cref{lem:rejection-rate},
and we implement this Boltzmann sampler
using \Cref{alg:multiset-sampler} and \Cref{alg:bec-subroutine}.
\Cref{lem:middle-rejection} ensures that the time and space complexities
of \Cref{alg:bec-subroutine} are linear in the size of the object produced.
\Cref{alg:multiset-sampler} runs in expected $O(n)$ time and space by
\Cref{prop:mset-sampler-correctness} and
our choice of $\lambda_n$.
Thus, \Cref{alg:bec-algorithm} runs in expected
$O(n^{1.625})$ time and uses expected $O(n)$ space.

To use deterministic $O(n)$ space, we modify the Boltzmann sampler to
reject partially constructed objects if their size is at least $2n$.
By Markov's inequality, this refined Boltzmann sampler outputs objects
of size less than $2n$ from a new Boltzmann distribution with
probability at least
\begin{align*}
  1 - \P_{\lambda_n}(U_n \ge 2n) &\ge 1 - \frac{\E_{\lambda_n}[U_n]}{2n}\\
  &= \frac{1}{2}.
\end{align*}
Thus, at most a constant number of trials are needed in expectation to
sample from the tail-truncated Boltzmann distribution. 
\end{proof}

The singularity analysis in Theorem~\ref{thm:new-main-theorem} generalizes to the
setting of weighted partitions parameterized by a positive integer sequence
of degree $r$, but we need to use a different Boltzmann sampler.
The truncated class~$\cC_n$ of weighted partitions has the generating function
\[
  C_n(z) = \prod_{k=1}^{n} \parens{1-z^k}^{-b_k},
\]
so we can use independent geometric random variables to sample the
number of parts each type of summand contributes to the final configuration.
See the Boltzmann samplers for the Cartesian product and
sequence operator in~\cite{FFP07} for more details.

\generaltheorem*

\begin{proof}
The tuning step is the same as in the proof of Theorem~\ref{thm:new-main-theorem}.
This Boltzmann algorithm samples from
  \[\sum_{k=1}^{n} b_k = O(n^{r+1})\]
geometric distributions, each taking time and space proportional to the
number they output. The total number of geometric trials across all $O(n^{r+1})$
distributions is $O(n)$ by our choice of $\lambda_n$.
Using Markov's inequality again, we guarantee that 
the algorithm uses deterministic $O(n)$ space.
\Cref{lem:rejection-rate} gives a rejection rate of $O(n^{3/4})$.
\end{proof}

\subsection{Rejection Rate of \Cref{alg:bec-subroutine}}
\label{app:bec-rejection-rates}

\Cref{lem:middle-rejection} shows that the rejection sampling in
\Cref{alg:bec-subroutine} takes a constant number of trials in expectation each
time it is called as a subroutine by \Cref{alg:multiset-sampler}.
We generalize our analysis from the case $d=3$ (Bose--Einstein condensates)
to $d \ge 1$ so that our arguments
involving negative binomial distributions will be useful in other contexts.
In particular, we develop a simple but effective tail inequality for negative
binomial random variables (\Cref{lem:d-root-convolution})
parameterized by high success probabilities that
outperforms the standard Chernoff-type inequality in this setting.

Throughout this subsection let $[n]=\{1,2,\dots,n\}$
and $V$ be a random variable for the size of an object drawn from $\Gamma\MSET_{\ge 1}(d\cZ)(\lambda)$.
The ordinary generating function for $\cB = \MSET(d\cZ)$ (the multisets of $d$ distinct atoms)
is
\[
  B(z) = \sum_{k=0}^\infty \binom{k+d-1}{d-1} z^k = \frac{1}{\parens*{1-z}^d},
\]
so we have $V \sim \NB_{\ge 1}(d, \lambda)$. 
Similarly, assume that $W \sim \NB(d, \lambda)$.

\begin{lemma}\label{lem:middle-rejection}
For $n$ sufficiently large and all $k \in [n]$, 
\[
  \P_{\lambda_n^k}(V \le n) \ge \frac{1}{2}.
\]
\end{lemma}

We use three lemmas to prove \Cref{lem:middle-rejection}.
First, \Cref{lem:x_n-case} shows that it suffices to lower bound the success probability
of \Cref{alg:bec-subroutine} for $k=1$ instead of all $k \in [n]$.
Recall that the probability mass function for $W$ is
\begin{equation}\label{eqn:nb-pmf}
  \P_{\lambda}(W=k) = \binom{k+d-1}{d-1}\lambda^k(1-\lambda)^d,
\end{equation}
The cumulative distribution function for $W$ is
\begin{equation}\label{eqn:nb-cdf}
  \P_{\lambda}(W \le k) = 1 - I_\lambda(k+1, d),
\end{equation}
where $I_\lambda(a, b)$ is the regularized incomplete beta function defined as
\[
  I_\lambda(a, b) \DEF \frac{B_\lambda(a, b)}{B_1(a, b)},
\]
with
\[
  B_\lambda(a, b) \DEF \int_{0}^\lambda t^{a-1} (1-t)^{b-1} \textrm{d}t.
\]



\begin{lemma}\label{lem:x_n-case}
For all $k \in [n]$, we have
\[
  \P_{\lambda_n^k}(V \le n) \ge \P_{\lambda_n}(V \le n).
\]
\end{lemma}
\begin{proof}
Using \Cref{eqn:nb-pmf} and \Cref{eqn:nb-cdf}, observe that 
\begin{align*}
  \P_{\lambda_n^k}(V \le n) &=
  \frac{\P_{\lambda_n^k}(W \le n) - \P_{\lambda_n^k}(W = 0)}{1 - \P_{\lambda_n^k}(W = 0)}\\
  &=\frac{(1-I_{\lambda_n^k}(n+1,d)) - (1-\lambda_n^k)^d}{1 - (1 - \lambda_n^k)^d}\\
  &= 1 - \frac{I_{\lambda_n^k}(n+1,d)}{1 - (1-\lambda_n^k)^d}.
\end{align*}
Since the integrand of the beta function
$t^{a-1}(1-t)^{b-1}$ is positive on $(0,1)$ and $0 < \lambda_n < 1$, we have
\[
  1 - \frac{I_{\lambda_n^k}(n+1,d)}{1 - (1-\lambda_n^k)^d}
  \ge  1 - \frac{I_{\lambda_n}(n+1,d)}{1 - (1-\lambda_n)^d},
\]
for all $k \in [n]$.
\end{proof}

\noindent
Thus, we only need to analyze the rejection rate when sampling
from the zero-truncated distribution $\NB_{\ge 1}(d, \lambda_n)$.

Second, \Cref{lem:x_n-bounds} gives an upper bound and lower bound for $\lambda_n$
using an asymptotic formula from the Khintchine--Meinardus probabilsitic method.

\begin{lemma}\label{lem:x_n-bounds}
For $\MSET(\MSET_{1..n}(d\cZ))$ and $n$ sufficiently large, we have
\[
  \exp\parens*{-2 n^{-\frac{1}{d+1}}} \le \lambda_n
 \le \exp\parens*{- \frac{1}{2} n^{-\frac{1}{d+1}}}.
\]
\end{lemma}

\begin{proof}
Let $\lambda_n = e^{-\delta_n}$.
Equation (32) in \cite{GS12} asserts that as $n \rightarrow \infty$,
\begin{equation*}
  \delta_n \sim \parens*{A_r \Gamma(\rho_r) \zeta(\rho_r + 1) \rho_r}^{\frac{1}{\rho_r + 1}}
                n^{-\frac{1}{\rho_r + 1}}.
\end{equation*}
In this instance, $r=d-1$ and $\rho_r = d$.
It follows that
$A_r = 1/(d-1)!$ and $\Gamma(\rho_r) = (d-1)!$, so
\[
  \delta_n \sim (\zeta(d+1) d)^{\frac{1}{d+1}} n^{-\frac{1}{d+1}},
\]
as $n \rightarrow \infty$.
Therefore, for any $\varepsilon > 0$ and $n$ sufficiently large, we have
\begin{align*}
  \delta_n \ge (1-\varepsilon)(\zeta(d+1) d)^{\frac{1}{d+1}} n^{-\frac{1}{d+1}},
\end{align*}
and
\begin{align*}
  \delta_n &\le  (1+\varepsilon)(\zeta(d+1) d)^{\frac{1}{d+1}} n^{-\frac{1}{d+1}}.
\end{align*}
The result follows from the fact $1 \le \zeta(d+1) \le \pi^2/6$ and by
letting $\varepsilon = 1/3$.
\end{proof}


Third, 
\Cref{lem:d-root-convolution} is a
tail inequality for negative binomial random variables.
Although our derivation is easily understood using standard techniques in 
enumerative combinatorics, the inequality captures an ample
amount of probability mass for all integers $n \ge 0$.
It is empirically much tighter than a Chernoff-type inequality in this setting as
$n \rightarrow \infty$ and $\lambda_n \rightarrow 1$.

\begin{lemma}[Negative Binomial Tail Inequality]
\label{lem:d-root-convolution}
Assume $W \sim \NB(d, \lambda)$. For all integers $n \ge 0$, 
\[
  \P(W > n) \le 1 - \parens*{1-\lambda^{n/d}}^d.
\]
\end{lemma}

\begin{proof}
Let $m = \floor{n/d}$ and observe that
\[
  \sum_{k=0}^n \binom{k+d-1}{d-1}\lambda^k \ge \parens*{\sum_{k=0}^{m} \lambda^k}^d.
\]
This inequality has a direct combinatorial interpretation in terms of weak
$d$-compositions. The left-hand side is the truncated ordinary generating
function for weak compositions of $k$ into $d$ parts, and the right-hand side
is the truncated ordinary generating function for weak compositions of $k$ into
$d$ parts of size at most $m$.
It follows that
\begin{align*}
  \parens*{\sum_{k=0}^{m} \lambda^k}^d &=
  \parens*{\frac{1-\lambda^{m+1}}{1-\lambda}}^d\\
  &\ge \parens*{1-\lambda^{n/d}}^d (1-\lambda)^{-d},
\end{align*}
since $0 < \lambda < 1$. Therefore, by \Cref{eqn:nb-pmf} we have
\begin{align*}
  \P(W > n) &= 1 - \P(W \le n)\\
   &= 1 - \sum_{k=0}^n \binom{k+d-1}{d-1}\lambda^k(1-\lambda)^d\\
  &\le 1 - \parens*{1-\lambda^{n/d}}^d,
\end{align*}
as desired.
\end{proof}

Combining the prerequisite lemmas and using the definition of the
probability mass function for the negative binomial distribution, we prove
\Cref{lem:middle-rejection}.

\begin{proof}[Proof of \Cref{lem:middle-rejection}.]
It suffices to consider the case when $k=1$ by \Cref{lem:x_n-case}.
By \Cref{lem:d-root-convolution} and \Cref{eqn:nb-pmf},
\begin{align*}
  \P_{\lambda_n}(V \le n) &= \frac{\P_{\lambda_n}(W \in [n])}{\P_{\lambda_n}(W \ge 1)}\\
     &\ge \P_{\lambda_n}(W \le n) - \P_{\lambda_n}(W = 0)\\
     &\ge \parens*{1-\lambda_n^{n/d}}^d - (1-\lambda_n)^d.
\end{align*}
Substituting the upper and lower bounds for $\lambda_n$ given in
\Cref{lem:x_n-bounds}, it follows that
\begin{align*}
  \parens*{1-\lambda_n^{n/d}}^d - (1-\lambda_n)^d
     &\ge \parens*{1 - \exp\parens*{-\frac{1}{2d}n^{\frac{d}{d+1}}}}^d\\
       &\phantom{\ge}\hspace{0.05cm} - \parens*{1-\exp\parens*{-2n^{-\frac{1}{d+1}}}}^d\\
     &\ge \frac{1}{2},
\end{align*}
for $n$ sufficiently large, because
\[
  \lim_{n\rightarrow\infty} \exp\parens*{-n^{\frac{d}{d+1}}} = 0,
\]
and
\[
  \lim_{n\rightarrow\infty} \exp\parens*{-n^{-\frac{1}{d+1}}} = 1.
\]
This completes the proof.
\end{proof}


\section{Conclusion}
We have shown how to analyze the complexity of Boltzmann samplers for a family of
weighted partitions (including Bose--Einstein condensates) through the
singularity analysis of Dirichlet generating functions.
In particular, we relate the degree of the polynomial parameterizing the
sequence~$(b_k)_{k=1}^\infty$ to the rejection rate of our algorithms through
a local limit theorem in~\cite{GS12}.
The main observation in our analysis is that the Dirichlet generating function
for a positive integer sequence of degree $r$ is a linear combination of
shifted Riemann zeta functions.
This allows us to use results from analytic number theory to
conveniently analyze the poles of these functions.
Other ideas in our analysis using a Newton interpolating polynomial to bound
residues, and developing a negative binomial tail inequality to analyze an
intermediate rejection rate.

Future directions of this work include analyzing these algorithms in
the interval or floating-point arithmetic models of computation,
instead of the real-arithmetic model.
The primary question to address is how accurate the approximation of
$\lambda_n$ must be in order to maintain a similar rejection rate,
since $\lambda_n$ approaches an essential singularity of
the generating function $C(z)$.

\section*{Acknowledgments}
We thank Marcel Celaya for various helpful discussions and
the anonymous reviewers of an earlier version of this paper for their insightful comments and suggestions.

\bibliographystyle{plain}
\bibliography{references}


\end{document}